\documentclass{article}
\usepackage{fullpage}
\usepackage{amsfonts}
\usepackage{graphicx}
\usepackage{amsmath,amssymb,amsthm}
\usepackage{authblk}

\newcommand{\ENTROPY}{\ensuremath{\mathsf{H}}}
\newcommand{\COMPLEXITY}{\ensuremath{\mathsf{C}}}
\newcommand{\ST}{\ensuremath{\mathsf{ST}}}
\newcommand{\SLT}{\ensuremath{\mathsf{SLT}}}
\newcommand{\KL}{\ensuremath{\mathsf{KL}}}
\newcommand{\Lemma}[1]{Lemma~\ref{#1}}
\newcommand{\ltdots}{..}
\newcommand{\SL}{\mathtt{suffixLink}}

\newcommand{\BWT}{\ensuremath{\mathsf{BWT}}}
\newcommand\SP[1]{\mathtt{sp}(#1)}
\newcommand\EP[1]{\mathtt{ep}(#1)} 
\newcommand\INTERVAL[1]{\mathtt{range}(#1)}

\newcommand{\repr}{\ensuremath{\mathtt{repr}}}

\newtheorem{theorem}{Theorem}

\newtheorem{corollary}{Corollary}

\begin{document} 

\title{A framework for space-efficient string kernels}
\author[1,2]{Djamal Belazzougui}
\author[1,2]{Fabio Cunial}
\affil[1]{Department of Computer Science, University of Helsinki, Finland.\thanks{This work was partially supported by Academy of Finland under grant 250345 (Center of Excellence in Cancer Genetics Research).}}
\affil[2]{Helsinki Institute for Information Technology, Finland.}
\maketitle

\maketitle

\begin{abstract}
String kernels are typically used to compare genome-scale sequences whose length makes alignment impractical, yet their computation is based on data structures that are either space-inefficient, or incur large slowdowns. We show that a number of exact string kernels, like the $k$-mer kernel, the substrings kernels, a number of length-weighted kernels, the minimal absent words kernel, and kernels with Markovian corrections, can all be computed in $O(nd)$ time and in $o(n)$ bits of space in addition to the input, using just a $\mathtt{rangeDistinct}$ data structure on the Burrows-Wheeler transform of the input strings, which takes $O(d)$ time per element in its output. The same bounds hold for a number of measures of compositional complexity based on multiple value of $k$, like the $k$-mer profile and the $k$-th order empirical entropy, and for calibrating the value of $k$ using the data.
%
\end{abstract}

\section{Introduction}

Given two strings $T^1$ and $T^2$, a \emph{kernel} is a function that simultaneously converts $T^1$ and $T^2$ into vectors $\mathbf{T}^1$ and $\mathbf{T}^2$ in $\mathbb{R}^n$ for some $n>0$, and computes a similarity or a distance measure between $\mathbf{T}^1$ and $\mathbf{T}^2$, \emph{without building and storing $\mathbf{T}^i$ explicitly} \cite{shawe2004kernel}. Kernels are often the method of choice for comparing extremely long strings, like genomes, read sets, and metagenomic samples, whose size makes alignment infeasible, yet their computation is typically based on space-inefficient data structures, like (truncated) suffix trees, or on space-efficient data structures with large slowdowns, like compressed suffix trees (see e.g. \cite{apostolico2010maximal,gog2011compressed} and references therein). The (possibly infinite) dimensions of $\mathbf{T}^i$ are, for example, all strings of a specific family on the alphabet of $T^1$ and $T^2$, and the value assigned to vector $\mathbf{T}^i$ along dimension $W$ corresponds to the number of occurrences of string $W$ in $T^i$, often rescaled and corrected in domain-specific ways. $\mathbf{T}^i$ is often called \emph{composition vector}, and a large number of its components can be zero in practice. In this paper we focus on space- and time-efficient algorithms for computing the \emph{cosine of the angle between two composition vectors $\mathbf{T}^1$ and $\mathbf{T}^2$}, i.e. on computing the kernel $\kappa(\mathbf{T}^1,\mathbf{T}^2) = N/\sqrt{D^1 D^2} \in [-1..1]$, where $N=\sum_{W}\mathbf{T}^1[W]\mathbf{T}^2[W]$ and $D^{i}=\sum_{W}\mathbf{T}^i[W]^2$. This measure of similarity can be converted into a distance $d(\mathbf{T}^1,\mathbf{T}^2)=(1-\kappa(\mathbf{T}^1,\mathbf{T}^2))/2 \in [0 \ltdots 1]$, and the algorithms we describe can be applied to compute norms of vector $\mathbf{T}^1-\mathbf{T}^2$, like the $p$-norm and the infinity norm. When $\mathbf{T}^1$ and $\mathbf{T}^2$ are bitvectors, we are more interested in interpreting them as sets and in computing the Jaccard distance $J(\mathbf{T}^1,\mathbf{T}^2) = ||\mathbf{T}^1 \wedge \mathbf{T}^2|| / || \mathbf{T}^1 \vee \mathbf{T}^2 || = ||\mathbf{T}^1 \wedge \mathbf{T}^2|| / ( ||\mathbf{T}^1|| + ||\mathbf{T}^2|| + ||\mathbf{T}^1 \wedge \mathbf{T}^2|| )$, where $\wedge$ and $\vee$ are the bitwise AND and OR operators, and where $|| \cdot ||$ measures the number of ones in a bitvector.

Given a data structure that supports $\mathtt{rangeDistinct}$ queries on the Burrows-Wheeler transform of each string in input, we show that a number of popular string kernels, like the $k$-mer kernel, the substrings kernels, a number of length-weighted kernels, the minimal absent words kernel, and kernels with Markovian corrections, can all be computed in $O(nd)$ time and in $o(n)$ bits of space in addition to the input, \emph{all in a single pass over the BWTs of the input strings}, where $d$ is the time taken by the $\mathtt{rangeDistinct}$ query per element in its output. The same bounds hold for computing a number of measures of compositional complexity \emph{for multiple values of $k$ at the same time}, like the $k$-mer profile and the $k$-th order empirical entropy, and for choosing the value of $k$ used in $k$-mer kernels from the data. All these algorithms become $O(n)$ using the $\mathtt{rangeDistinct}$ data structure described in \cite{BNV13}, and concatenating this setup to the BWT construction algorithm described in \cite{Belazzougui14}, we can compute all such kernels and complexity measures \emph{from the input strings} in randomized $O(n)$ time and in $O(n\log{\sigma})$ bits of space in addition to the input. Finally, we show that measures of expectation based on Markov models are related to the left and right extensions of maximal repeats.

\section{Preliminaries}

\subsection{Strings}

Let $\Sigma=[1..\sigma]$ be an integer alphabet, let $\#=0$, $\#_1=-1$ and $\#_2=-2$ be distinct separators not in $\Sigma$, and let $T=[1..\sigma]^{n-1}\#$ be a string. We assume $\sigma \in o(\sqrt{n}/\log{n})$ throughout the paper. A \emph{$k$-mer} is any string $W \in [1..\sigma]$ of length $k>0$. We denote by $f_{T}(W)$ the number of (possibly overlapping) occurrences of a string $W$ in the circular version of $T$, and we use the shorthand $p_{T}(W)=f_{T}(W)/(n-|W|)$ to denote an approximation of the \emph{empirical probability} of observing $W$ in $T$, assuming that all positions of $T$ except the last $|W|$ ones are equally probable starting positions for $W$. A \emph{repeat} $W$ is a string that satisfies $f_{T}(W)>1$. We denote by $\Sigma^{\ell}_{T}(W)$ the set of characters $\{a \in [0..\sigma] : f_{T}(aW)>0\}$ and by $\Sigma^{r}_{T}(W)$ the set of characters $\{b \in [0..\sigma] : f_{T}(Wb)>0\}$. A repeat $W$ is \emph{right-maximal} (respectively, \emph{left-maximal}) iff $|\Sigma^{r}_{T}(W)|>1$ (respectively, iff $|\Sigma^{\ell}_{T}(W)|>1$). It is well known that $T$ can have at most $n-1$ right-maximal substrings and at most $n-1$ left-maximal substrings. A \emph{maximal repeat} of $T$ is a repeat that is both left- and right-maximal.

For reasons of space we assume the reader to be familiar with the notion of \emph{suffix tree} $\ST_T$ of a string $T$, and with the notion of \emph{generalized suffix tree} of two strings, which we do not define here. 
We denote by $\ell(v)$ the string label of a node $v$ in a suffix tree. It is well known that a substring $W$ of $T$ is right-maximal iff $W=\ell(v)$ for some internal node $v$ of $\ST_T$. We assume the reader to be familiar with the notion of \emph{suffix link} connecting a node $v$ with $\ell(v)=aW$ for some $a \in [0..\sigma]$ to a node $w$ with $\ell(w)=W$: we say that $w=\SL(v)$ in this case. Here we just recall that suffix links and internal nodes of $\ST_T$ form a tree, called the \emph{suffix-link tree} of $T$ and denoted by $\SLT_T$, and that inverting the direction of all suffix links yields the so-called \emph{explicit Weiner links}. Given an internal node $v$ and a symbol $a \in [0..\sigma]$, it might happen that string $a\ell(v)$ does occur in $T$, but that it is not right-maximal, i.e. it is not the label of any internal node of $\ST_T$: all such left extensions of internal nodes that end in the middle of an edge are called \emph{implicit Weiner links}. An internal node $v$ of $\ST_T$ can have more than one outgoing Weiner link, and all such Weiner links have distinct labels: in this case, $\ell(v)$ is a maximal repeat. 
%
It is known that the number of suffix links (or, equivalently, of explicit Weiner links) is upper-bounded by $2n-2$, and that the number of implicit Weiner links can be upper-bounded by $2n-2$ as well.

\subsection{Enumerating right-maximal substrings and maximal repeats}\label{sec:enumerator}

For reasons of space we assume the reader to be familiar with the notion and uses of the Burrows-Wheeler transform of $T$, including the $C$ array, the $\mathtt{rank}$ function, and backward searching. In this paper we use $\BWT_T$ to denote the BWT of $T$, we use $\INTERVAL{W} = [\SP{W}..\EP{W}]$ to denote the lexicographic interval of a string $W$ in a BWT that is implicit from the context, and we use $\Sigma_{i,j}$ to denote the set of distinct characters that occur inside interval $[i..j]$ of a string that is implicit from the context. We also denote by $\mathtt{rangeDistinct}(i,j)$ the function that returns the set of tuples $\{(c,\mathtt{rank}(c,p_c),\mathtt{rank}(c,q_c)) : c \in \Sigma_{i,j} \}$, where $p_c$ and $q_c$ are the first and the last occurrence of $c$ inside interval $[i..j]$, respectively. Here we focus on a specific application of $\BWT_T$: enumerating all the right-maximal substrings of $T$, or equivalently all the internal nodes of $\ST_T$. In particular, we use the algorithm described in \cite{Belazzougui14} (Section 4.1), which we sketch here for completeness.

Given a substring $W$ of $T$, let $b_1 < b_2 < \dots < b_k$ be the sorted sequence of all the distinct characters in $\Sigma^{r}(W)$, and let $a_1,a_2,\dots,a_h$ be the list of all the characters in $\Sigma^{\ell}(W)$, not necessarily sorted. Assume that we represent a substring $W$ of $T$ as a pair $\repr(W)=(\mathtt{chars}[1..k],\mathtt{first}[1..k+1])$, where $\mathtt{chars}[i]=b_i$, $\INTERVAL{Wb_i}=[\mathtt{first}[i]..\mathtt{first}[i+1]-1]$ for $i \in [1..k]$, and $\INTERVAL$ refers to $\BWT_T$. Note that $\INTERVAL{W}=[\mathtt{first}[1]..\mathtt{first}[k+1]-1]$, since it coincides with the concatenation of the intervals of the right extensions of $W$ in lexicographic order. If $W$ is not right-maximal, array $\mathtt{chars}$ in $\repr(W)$ has length one. Given a data structure that supports $\mathtt{rangeDistinct}$ queries on $\BWT_T$, and given the $C$ array of $T$, there is an algorithm that converts $\repr(W)$ into the sequence $a_1,\dots,a_h$ and into the corresponding sequence $\repr(a_{1}W),\dots,\repr(a_{h}W)$, in $O(de)$ time and $O(\sigma^{2}\log{n})$ bits of space in addition to the input and the output \cite{Belazzougui14}, where $d$ is the time taken by the $\mathtt{rangeDistinct}$ operation per element in its output, and $e$ is the number of distinct strings $a_{i}Wb_{j}$ that occur in the circular version of $T$, where $i \in [1..h]$ and $j \in [1..k]$. We encapsulate this algorithm into a function that we call $\mathtt{extendLeft}$.

If $a_i W$ is right-maximal, i.e. if array $\mathtt{chars}$ in $\repr(a_{i}W)$ has length greater than one, we push pair $(\repr(a_i W),|W|+1)$ onto a stack $S$. In the next iteration we pop the representation of a string from the stack and we repeat the process, until the stack itself becomes empty. This process is equivalent to following all the explicit Weiner links from the node $v$ of $\ST_T$ with $\ell(v)=W$, not necessarily in lexicographic order. Thus, running the algorithm from a stack initialized with $\repr(\varepsilon)$ is equivalent to performing a depth-first (but not necessarily a preorder) traversal of the suffix-link tree of $T$, which guarantees to enumerate all the right-maximal substrings of $T$. 
Every operation performed by the algorithm can be charged to a distinct node or Weiner link of $\ST_T$, thus the algorithm runs in $O(nd)$ time. The depth of the stack is $O(\log{n})$ rather than $O(n)$, since at every iteration we push the pair $(\repr(a_i W),|a_i W|)$ with largest $\INTERVAL{a_i W}$ first. Every suffix-link tree level in the stack contains at most $\sigma$ pairs, and each pair takes at most $\sigma\log{n}$ bits of space, thus the total space used by the stack is $O(\sigma^2 \log^{2}{n})$ bits. The following theorem follows from our assumption that $\sigma \in o(\sqrt{n}/\log{n})$:

\begin{theorem}[\cite{Belazzougui14}]\label{thm:enumerator}
Let $T \in [1..\sigma]^{n-1}\#$ be a string. Given a data structure that supports $\mathtt{rangeDistinct}$ queries on $\BWT_T$, we can enumerate all the right-maximal substrings $W$ of $T$, and for each of them we can return $|W|$, $\repr(W)$, the sequence $a_1,a_2,\dots,a_h$ of all characters in $\Sigma_{T}^{\ell}(W)$ (not necessarily sorted), and the sequence $\repr(a_{1}W),\dots,\repr(a_{h}W)$, in $O(nd)$ time and in $o(n)$ bits of space in addition to the input and the output, where $d$ is the time taken by the $\mathtt{rangeDistinct}$ operation per element in its output.
\end{theorem}

Theorem \ref{thm:enumerator} does not specify the order in which the right-maximal substrings must be enumerated, nor the order in which the left extensions of a right-maximal substring must be returned. The algorithm we just described can be adapted to return all the maximal repeats of $T$, with the same bounds, by outputting a right-maximal string $W$ iff $|\mathtt{rangeDistinct}(\SP{W},\EP{W})|>1$. A version of the same algorithm can also enumerate all the internal nodes of the \emph{generalized suffix tree} of two string $T^1$ and $T^2$, using $\BWT_{T^1}$ and $\BWT_{T^2}$: in this case, a string $W$ is represented as a quadruple $\repr'(W)=(\mathtt{chars}_{1}[1..k_1],\mathtt{first}_{1}[1..k_1+1],\mathtt{chars}_{2}[1..k_2],\mathtt{first}_{2}[1..k_2+1])$, and we assume that $\mathtt{first}_{i}[1]=0$ iff $W$ does not occur in $T^i$. We call $\mathtt{extendLeft}'$ the function that maps $\repr'(W)$ to the list of its left extensions $\repr'(a_{i}W)$.

\begin{theorem}[\cite{Belazzougui14}]\label{thm:generalizedEnumerator}
Let $T^1 \in [1..\sigma]^{n_1-1}\#_1$ and $T^2 \in [1..\sigma]^{n_2-1}\#_2$ be two strings. Given two data structures that support $\mathtt{rangeDistinct}$ queries on $\BWT_{T^1}$ and on $\BWT_{T^2}$, respectively, we can enumerate all the right-maximal substrings $W$ of $T=T^1 T^2$, and for each of them we can return $|W|$, $\repr'(W)$, the sequence $a_1,a_2,\dots,a_h$ of all characters in $\Sigma_{T^{1}T^2}^{\ell}(W)$ (not necessarily sorted), and the sequence $\repr'(a_{1}W),\dots,\repr'(a_{h}W)$, in $O(nd)$ time and in $o(n)$ bits of space in addition to the input and the output, where $n=n_1+n_2$ and $d$ is the time taken by the $\mathtt{rangeDistinct}$ operation per element in its output.
\end{theorem}

For reasons of space, we assume throughout the paper that $d$ is the time per element in the output of a $\mathtt{rangeDistinct}$ data structure that is implicit from the context. We also replace $T^i$ by $i$ in subscripts, or we waive subscripts completely whenever they are clear from the context.

\section{Kernels and complexity measures on $k$-mers}\label{sect:kmerKernels}

Given a string $T \in [1..\sigma]^{n-1}\#$ and a length $k>0$, let vector $\mathbf{T}_{k}=[1..\sigma^k]$ be such that $\mathbf{T}_{k}[W]=f_{T}(W)$ for every $W \in [1..\sigma]^k$. The \emph{$k$-mer complexity} $\COMPLEXITY_{k}(T)$ of string $T$ is the number of nonzero components of $\mathbf{T}_{k}$. The \emph{$k$-mer kernel} of two strings $T^1$ and $T^2$ is $\kappa(\mathbf{T}^{1}_{k},\mathbf{T}^{2}_{k})$. Recall that Theorem \ref{thm:enumerator} and \ref{thm:generalizedEnumerator} enumerate all nodes of a suffix tree in no specific order. In this section we describe algorithms to compute $\COMPLEXITY_{k}(T)$ and $\kappa(\mathbf{T}^1_k,\mathbf{T}^2_k)$ in a way that does not depend on the order in which the nodes of a suffix tree are enumerated: we can thus implement such algorithms on top of Theorem \ref{thm:enumerator} and \ref{thm:generalizedEnumerator}. The main idea behind our approach is a telescoping strategy that works by adding and subtracting terms in a sum, as described below:

\begin{theorem} \label{thm:kmerComplexity}
Let $T \in [1..\sigma]^{n-1}\#$ be a string. Given an integer $k$ and a data structure that supports $\mathtt{rangeDistinct}$ queries on $\BWT_T$, we can compute $\COMPLEXITY_{k}(T)$ in $O(nd)$ time and in $o(n)$ bits of space in addition to the input.
\end{theorem}
\begin{proof}
A $k$-mer of $T$ can either be the label of a node of $\ST_T$, or it could end in the middle of an edge $(u,v)$ of $\ST$. In the latter case, we assume that the $k$-mer is represented by its locus $v$, which might be a leaf. Let $\COMPLEXITY_{k}(T)$ be initialized to $n-k$, i.e. to the number of leaves that correspond to suffixes of $T$ of length at least $k+1$. We enumerate the internal nodes of $\ST$ using Theorem \ref{thm:enumerator}, and every time we enumerate a node $v$ we proceed as follows: if $|\ell(v)|<k$ we leave $\COMPLEXITY_{k}(T)$ unaltered, otherwise we increment $\COMPLEXITY_{k}(T)$ by one and we decrement $\COMPLEXITY_{k}(T)$ by the number of children of $v$ in $\ST$, which is the length of array $\mathtt{chars}$ in $\repr(\ell(v))$. In this way, every internal node $v$ of $\ST$ that is located at string depth at least $k$ and that is not the locus of a $k$-mer is both added to $\COMPLEXITY_{k}(T)$ (when the algorithm visits $v$) and subtracted from $\COMPLEXITY_{k}(T)$ (when the algorithm visits $\mathtt{parent}(v)$). Leaves at depth at least $k+1$ that are not the locus of a $k$-mer are added by the initialization of $\COMPLEXITY_{k}(T)$, and they are subtracted during the enumeration. Conversely, every locus $v$ of a $k$-mer of $T$ (including leaves) is just added to $\COMPLEXITY_{k}(T)$, since $|\ell(\mathtt{parent}(v))|<k$.
\end{proof}

We can apply the same telescoping strategy to compute $\kappa(\mathbf{T}^{1}_{k},\mathbf{T}^{2}_{k})$:

\begin{theorem}\label{thm:kmerKernel}
Let $T^1 \in [1..\sigma]^{n_1-1}\#_1$ and $T^2 \in [1..\sigma]^{n_2-1}\#_2$ be strings. Given an integer $k$ and two data structures that support $\mathtt{rangeDistinct}$ queries on $\BWT_{T^1}$ and on $\BWT_{T^2}$, respectively, we can compute $\kappa(\mathbf{T}^1_k,\mathbf{T}^2_k)$ in $O(nd)$ time and in $o(n)$ bits of space in addition to the input, where $n=n_1+n_2$.
\end{theorem}
\begin{proof}
Recall that $\kappa(\mathbf{T}_k^1,\mathbf{T}_k^2) = N/\sqrt{D^1 D^2}$, where $N=\sum_{W}\mathbf{T}_k^1[W]\mathbf{T}_k^2[W]$, $D^{i}=\sum_{W}\mathbf{T}_k^i[W]^2$, and $W \in [1..\sigma]^k$. We initially set $N=0$ and $D^i=n_i-k$, since these are the contributions of all the leaves at depth at least $k+1$ in the generalized suffix tree of $T^1$ and $T^2$. Then, we enumerate every internal node $u$ of the generalized suffix tree, using Theorem \ref{thm:generalizedEnumerator}: if $|\ell(u)|<k$  we keep all variables unchanged, otherwise we set $N$ to $N + f_{1}(\ell(u)) \cdot f_{2}(\ell(u)) - \sum_{v} f_{1}(\ell(v)) \cdot f_{2}(\ell(v))$ and we set $D^i$ to $D^i + f_{i}(\ell(u))^2 - \sum_{v}f_{i}(\ell(v))^2$, where $v$ ranges over all children of $u$ in the generalized suffix tree. Clearly $f_{i}(\ell(u)) = \mathtt{first}_{i}[k_i+1]-\mathtt{first}_{i}[1]$ where $k_i$ is the size of array $\mathtt{chars}_i$ in $\repr'(\ell(u))$, and $f_{i}(\ell(v)) = f_{i}(\ell(u)b_j) = \mathtt{first}_{i}[j+1]-\mathtt{first}_{i}[j]$ for some $j \in [1..k_i]$. In analogy to \Lemma{thm:kmerComplexity}, the contribution of the loci of the distinct $k$-mers of $T^1$, of $T^2$, or of both, is added to the three temporary variables and never subtracted, while the contribution of every other node $u$ at depth at least $k$ in the generalized suffix tree is both added (when the algorithm visits $u$, or when $N$ and $D^i$ are initialized) and subtracted (when the algorithm visits $\mathtt{parent}(u)$).
\end{proof}

An even more specific notion of compositional complexity is $\COMPLEXITY_{k,f}(T)$, the number of distinct $k$-mers that occur exactly $f$ times in $T$. In the \emph{$k$-mer profiling} problem \cite{chikhi2013informed,chor2009genomic} we are given a string $T$, an interval $[k_1..k_2]$ of lengths and an interval $[f_1..f_2]$ of frequencies, and we are asked to compute the matrix $\mathtt{profile}[k_1..k_2,f_1..f_2]$ defined as follows: $\mathtt{profile}[i,j] = \COMPLEXITY_{i,j}(T)$ if $j<f_2$, and $\mathtt{profile}[i,j] = \sum_{h \geq j}\COMPLEXITY_{i,h}(T)$ if $j=f_2$. Note that column $j$ of $\mathtt{profile}$ can have nonzero cells only if $f_j$ is the frequency of some internal node of $\ST_T$. In practice $\mathtt{profile}$ is often computed by running a $k$-mer extraction algorithm $k_2-k_1+1$ times, and by scanning the output of all such runs (see e.g. \cite{chikhi2013informed} and references therein). The following lemma shows that we can compute $\mathtt{profile}$ in just one pass over the BWT of the input string, and in linear time in the size of $\mathtt{profile}$:

\begin{theorem}
Let $T \in [1..\sigma]^{n-1}\#$ be a string. Given ranges $[k_1..k_2]$ and $[f_1..f_2]$, and given a data structure that supports $\mathtt{rangeDistinct}$ queries on $\BWT_T$, we can compute matrix $\mathtt{profile}[k_1..k_2,f_1..f_2]$ in $O(nd+(k_2-k_1)(f_2-f_1))$ time and in $o(n)$ bits of space in addition to the input and the output.
\end{theorem}
\begin{proof}
We use Theorem \ref{thm:enumerator} again. Assume that, for every internal node $u$ of $\ST_T$ with string depth at least $k_1$ and with frequency at least $f_1$, and for every $k \in [k_1..\min\{|\ell(v)|,k_2\}]$, we increment $\mathtt{profile}[k,\min\{f(u),f_2\}]$ by one and we decrement $\mathtt{profile}[k,\min\{f(v),f_2\}]$ by one for every child $v$ of $u$ in $\ST$ such that $f(v) \geq f_1$. This would take $O(n^2)$ total updates to $\mathtt{profile}$. However, we can perform all of these updates in batch, as follows: for every node $u$ of $\ST$ with $f(u) \geq f_1$ and with $|\ell(u)| \geq k_1$, we just increment $\mathtt{profile}[\min\{|\ell(u)|,k_2\}$, $\min\{f(u),f_2\}]$ by one, and we just decrement $\mathtt{profile}[\min\{|\ell(u)|,k_2\}$, $\min\{f(v),f_2\}]$ by one for every child $v$ of $u$ in $\ST$ such that $f(v) \geq f_1$. After having traversed all the internal nodes of $\ST$, we scan $\mathtt{profile}$ as follows: for every $j \in [f_1..f_2]$, we traverse all values of $i$ in the decreasing order $k_2-1,\dots,k_1$, and we set $\mathtt{profile}[i,j]=\mathtt{profile}[i,j]+\mathtt{profile}[i+1,j]$. If $f_1=1$, at the end of this process the first column of $\mathtt{profile}$ contains negative numbers, since Theorem \ref{thm:enumerator} does not enumerate the leaves of $\ST$. Thus, before returning, we add to $\mathtt{profile}[i,1]$ the number of leaves with string depth at least $k_i+1$, i.e. value $n-k_i$, for all $i \in [k_1..k_2]$.
\end{proof}

A similar algorithm allows computing $\kappa(\mathbf{T}^1_k,\mathbf{T}^2_k)$ for all $k$ in a user-specified range $[k_1..k_2]$ in $O(nd+k_2-k_1)$ time. Matrix $\mathtt{profile}$ can be used to determine a range of values of $k$ to be used in $k$-mer kernels. The smallest number in this range is typically the value of $k$ that maximizes the number of distinct $k$-mers that occur at least twice in $T$ \cite{sims2009alignment}. The largest number in the range is typically determined using some measure of expectation: we cover this computation in Section \ref{sec:markovianCorrection}.

A related notion of compositional complexity is the \emph{$k$-th order empirical entropy} of $T$, defined as $\ENTROPY_{k}(T) = (1/|T|) \cdot \sum_{W}\sum_{a \in \Sigma^{r}(W)} f_{T}(Wa) \cdot \log(f_{T}(W)/f_{T}(Wa))$, where $W$ ranges over all strings in $[1..\sigma]^k$. Clearly only the internal nodes of $\ST_T$ contribute to some $\ENTROPY_{k}(T)$ \cite{gog2011compressed}, thus our methods allow computing $\ENTROPY_{k}(T)$ for a user-specified \emph{range of lengths} $[k_1..k_2]$ in $O(nd)$ time, using just one pass over $\BWT_T$.


\section{Kernels and complexity measures on all substrings}

Given a string $T \in [1..\sigma]^{n-1}\#$, consider the infinite-dimensional vector $\mathbf{T}_{\infty}$ indexed by all distinct substrings $W \in [1..\sigma]^+$, such that $\mathbf{T}_{\infty}[W]=f_{T}(W)$. The \emph{substring complexity} $\COMPLEXITY_{\infty}(T)$ of $T$ is the number of nonzero components of $\mathbf{T}_{\infty}$. The \emph{substring kernel} of two strings $T^1$ and $T^2$ is the cosine of composition vectors $\mathbf{T}_{\infty}^1$ and $\mathbf{T}_{\infty}^2$. Computing substring complexity and substring kernel amounts to applying the same telescoping strategy described in Theorem \ref{thm:kmerComplexity} and \ref{thm:kmerKernel}, but with different contributions: 

\begin{corollary} \label{cor:substringComplexity}
Let $T \in [1..\sigma]^{n-1}\#$ be a string. Given a data structure that supports $\mathtt{rangeDistinct}$ queries on $\BWT_T$, we can compute $\COMPLEXITY_{\infty}(T)$ in $O(nd)$ time and in $o(n)$ bits of space in addition to the input.
\end{corollary}
\begin{proof}
The substring complexity of $T$ coincides with the number of characters in $[1..\sigma]$ that occur on all edges of $\ST_{T}$. We can thus proceed as in Lemma \ref{thm:kmerComplexity}, initializing $\COMPLEXITY_{\infty}(T)$ to $(n-1)n/2$, or equivalently to the sum of the lengths of all suffixes of $T[1..n-1]$. Whenever we visit a node $v$ of $\ST$, we add to $\COMPLEXITY_{\infty}(T)$ the quantity $|\ell(v)|$, and we subtract from $\COMPLEXITY_{\infty}(T)$ the quantity $|\ell(v)| \cdot |\mathtt{children}(v)|$. The net effect of all such operations coincides with summing the lengths of all edges of $\ST$, discarding all occurrences of character $\#$. Note that $|\ell(u)|$ is provided by Theorem \ref{thm:enumerator}, and $|\mathtt{children}(v)|$ is the size of array $\mathtt{chars}$ in $\repr(\ell(v))$.
\end{proof}

\begin{corollary} \label{cor:substringKernel}
Let $T^1 \in [1..\sigma]^{n_1-1}\#_1$ and $T^2 \in [1..\sigma]^{n_2-1}\#_2$ be strings. Given data structures that support $\mathtt{rangeDistinct}$ queries on $\BWT_{T^1}$ and on $\BWT_{T^2}$, respectively, we can compute $\kappa(\mathbf{T}_{\infty}^1,\mathbf{T}_{\infty}^2)$ in $O(nd)$ time and in $o(n)$ bits of space in addition to the input, where $n=n_1+n_2$.
\end{corollary}
\begin{proof}
We proceed as in Theorem \ref{thm:kmerKernel}, setting again $N=0$ and $D^i=(n_{i}-1)n_{i}/2$ at the beginning of the algorithm. When we visit a node $u$ of the generalized suffix tree of $T^1$ and $T^2$, we set $N$ to $N + |\ell(u)| \cdot ( f_{1}(\ell(u)) f_{2}(\ell(u)) - \sum_{v}f_{1}(\ell(v)) f_{2}(\ell(v)) )$ and we set $D^i$ to $D^i + |\ell(u)| \cdot ( f_{i}(\ell(u))^2 - \sum_{v} f_{i}(\ell(v))^2 )$, where $v$ ranges over all children of $u$ in the generalized suffix tree.
\end{proof}

In a substring kernel it is common to weight a substring $W$ by a user-specified function of its length: typical choices are $\epsilon^{|W|}$ for a given constant $\epsilon$, or indicators that select only substrings within a specific range of lengths \cite{smola2003fast}. We denote by $\mathbf{T}_{\infty,g}^i$ a weighted version of the infinite-dimensional vector $\mathbf{T}_{\infty}^i$ in which $\mathbf{T}_{\infty}^{i}[W] = g(|W|) \cdot f_{T^i}(W)$ and where $g$ is any user-specified function. We assume that the number of bits required to represent the output of $g$ with sufficient precision is $O(\log{n})$. It is easy to adapt Corollary \ref{cor:substringKernel} to support this type of composition vector:

\begin{corollary}\label{cor:weightedSubstringKernel}
Let $T^1 \in [1..\sigma]^{n_1-1}\#_1$ and $T^2 \in [1..\sigma]^{n_2-1}\#_2$ be strings. Given a function $g(k)$ that can be evaluated in constant time, and given data structures that support $\mathtt{rangeDistinct}$ queries on $\BWT_{T^1}$ and on $\BWT_{T^2}$, respectively, we can compute $\kappa(\mathbf{T}_{\infty,g}^1,\mathbf{T}_{\infty,g}^2)$ in $O(nd)$ time and in $o(n)$ bits of space in addition to the input, where $n=n_1+n_2$.
\end{corollary}
\begin{proof}
We modify Corollary \ref{cor:substringKernel} as follows. Assume that we are processing an internal node $v$ of the generalized suffix tree, let $\ell(v)=W$, and assume that we have computed $\repr'(aW)$ for all the left extensions $aW$ of $W$. In addition to pushing $\repr'(aW)$ onto the stack, we also push value $\mathtt{prefixSum}(aW)=\sum_{i=1}^{|W|+1}g(i)^2$ with it, where $\mathtt{prefixSum}(aW) = \mathtt{prefixSum}(W)+g(|W|+1)^2$. When we pop $\repr'(aW)$, we compute its contributions to $N$ and $D^i$ as described in Corollary \ref{cor:substringKernel}, but replacing $|aW|$ by $\mathtt{prefixSum}(aW)$. We initialize $D^i$ to $\sum_{j=1}^{n_{i}-1}g(j)^2$.
\end{proof}

Corollary \ref{cor:weightedSubstringKernel} can clearly support distinct weight functions for $T^1$ and $T^2$. For some functions, like $\epsilon^{|W|}$, prefix sums can be computed in closed form \cite{smola2003fast}, thus there is no need to push $\mathtt{prefixSum}$ values on the stack. Another frequent weighting scheme for a string $W$ associates a score $q(c)$ to every character $c$ of $W$, and it weights $W$ by e.g. $q(W)=\prod_{i=1}^{|W|}q(W[i])$. In this case we could just push $\mathtt{prefixSum}(V) = \sum_{i=1}^{|V|} \prod_{j=1}^{i}q(V[j])^2$ onto the stack, where $V=aW$ and $\mathtt{prefixSum}(V) = q(a)^2 \cdot (1 + \mathtt{prefixSum}(W))$. A similar weighting scheme can be used for $k$-mers as well. Let $\mathbf{T}_{k,q}$ be a version of $\mathbf{T}_k$ such that $\mathbf{T}_{k,q}[W]=f_{T}(W)-(|T|-|W|)q(W)$ for every $W \in [1..\sigma]^k$, and consider the following distances defined in \cite{reinert2009alignment}:
\begin{eqnarray*}
D_{2}^{s}(\mathbf{T}^1_{k,q},\mathbf{T}^2_{k,q}) & = & \sum_{W} \mathbf{T}^1_{k,q}[W]\mathbf{T}^2_{k,q}[W] / \sqrt{(\mathbf{T}^1_{k,q}[W])^2 + (\mathbf{T}^2_{k,q}[W])^2} \\
D_{2}^{*}(\mathbf{T}^1_{k,q},\mathbf{T}^2_{k,q}) & = & \sum_{W} \mathbf{T}^1_{k,q}[W]\mathbf{T}^2_{k,q}[W] / \left( \sqrt{(n_1-k)(n_2-k)} \cdot q(W) \right)
\end{eqnarray*}
where $W$ ranges over all strings in $[1..\sigma]^k$. We can compute such distances using just a minor modification to Theorem \ref{thm:kmerKernel}:

\begin{corollary}
Let $T^1 \in [1..\sigma]^{n_1-1}\#_1$ and $T^2 \in [1..\sigma]^{n_2-1}\#_2$ be strings. Given an integer $k$ and data structures that support $\mathtt{rangeDistinct}$ queries on $\BWT_{T^1}$ and on $\BWT_{T^2}$, respectively, we can compute $D_{2}^{s}(\mathbf{T}^1_{k,p},\mathbf{T}^2_{k,p})$ and $D_{2}^{*}(\mathbf{T}^1_{k,p},\mathbf{T}^2_{k,p})$ in $O(nd)$ time and in $\lambda\log{\sigma}+o(n)$ bits of space in addition to the input, where $n=n_1+n_2$ and $\lambda$ is the length of the longest repeat in $T^{1}T^2$.
\end{corollary}
\begin{proof}
We proceed as in Theorem \ref{thm:kmerKernel}, pushing on the stack value $q(W,k)=\prod_{j=1}^{k}q(W[j])$ in addition to $\repr'(W)$, and maintaining a separate stack of characters to represent the string we are processing during the depth-first traversal of the generalized suffix-link tree. We set $q(aW,k)=q(a) \cdot q(W,k) / q(b)$, where $b$ is the $k$th character from the top of the character stack when we are processing $W$.
\end{proof}

An orthogonal way to measure the similarity between $T^1$ and $T^2$ consists in comparing the repertoire of all strings that \emph{do not appear} in $T^1$ and in $T^2$. Given a string $T$ and two frequency thresholds $\tau_1 < \tau_2$, a string $W$ is a \emph{minimal rare word} of $T$ if $\tau_1 \leq f_{T}(W) < \tau_2$ and if $f_{T}(V) \geq \tau_2$ for every proper substring $V$ of $W$. Setting $\tau_1=0$ and $\tau_2=1$ gives the well-known \emph{minimal absent words} (see e.g. \cite{herold2008efficient,chairungsee2012using} and references therein), whose total number can be $\Theta(\sigma n)$ \cite{crochemore1998automata}. Setting $\tau_1=1$ and $\tau_2=2$ gives the so-called \emph{shortest unique substrings} (see e.g. \cite{ileri2014shortest} and references therein), whose total number is $O(n)$, like the number of strings obtained by any other setting of $\tau_1 \geq 1$. In what follows we focus on minimal absent words, but our algorithms can be generalized to other settings of the thresholds.

To decide whether $aWb$ is a minimal absent word of $T$, where $a$ and $b$ are characters, it clearly suffices to check whether $f_{T}(aWb)=0$ and whether both $f_{T}(aW) \geq 1$ and $f_{T}(Wb) \geq 1$. It is well known that only a maximal repeat of $T$ can be the infix $W$ of a minimal absent word $aWb$, and this applies to any setting of $\tau_1$ and $\tau_2$. To enumerate all the minimal absent words, for example to count their total number $\COMPLEXITY_{-}(T)$, we can thus iterate over all nodes of $\ST_T$ associated with maximal repeats, as described below:

\begin{theorem}\label{thm:maw}
Let $T \in [1..\sigma]^{n-1}\#$ be a string. Given a data structure that supports $\mathtt{rangeDistinct}$ queries on $\BWT_T$, we can compute $\COMPLEXITY_{-}(T)$ in $O(nd)$ time and in $o(n)$ bits of space in addition to the input.
\end{theorem}
\begin{proof}
For clarity, we first describe how to \emph{enumerate} all the distinct minimal absent words of $T$: we specialize this algorithm to counting at the end of the proof. We use Theorem \ref{thm:enumerator} to enumerate all nodes $v$ of $\ST_T$ associated with maximal repeats, as described in Section \ref{sec:enumerator}. Let $\{a_1,\dots,a_h\}$ be the set of distinct left extensions of string $\ell(v)$ in $T$ returned by operation $\mathtt{extendLeft}(\repr(v))$, let $\mathtt{extensions}[1..\sigma+1,0..\sigma]$ be a boolean matrix initialized to all zeros, and let $\mathtt{leftExtensions}[1..\sigma+1]$ be an array initialized to all zeros. Let $h'$ be a pointer initialized to one. Operation $\mathtt{extendLeft}$ allows following all the Weiner links from $v$, not necessarily in lexicographic order: for every string $a_i \ell(v)$ obtained in this way, we set $\mathtt{leftExtensions}[h']=a_i$, we enumerate its right extensions $\{c_1,\dots,c_{k'}\}$ using array $\mathtt{chars}$ of $\repr(a_{i}\ell(v))$, we set $\mathtt{extensions}[h',c_j]=1$ for all $j \in [1..k']$, and we finally increment $h'$ by one. Note that only the columns of $\mathtt{extensions}$ that correspond to the right extensions of $\ell(v)$ are updated by this procedure. Then, we enumerate all the right extensions $\{b_1,\dots,b_k\}$ of $\ell(v)$ using array $\mathtt{chars}$ of $\repr(\ell(v))$, and for every such extension $b_j$ we report all pairs $(a_i,b_j)$ such that $a_i=\mathtt{chars}[x]$, $x \in [1..h']$, and $\mathtt{extensions}[x,b_j]=0$. This process takes time proportional to the number of Weiner links from $v$, plus the number of children of $v$, plus the number of Weiner links from $v$ multiplied by $\sigma$. When applied to all nodes of $\ST$, this takes in total $O(n\sigma)$ time, which is optimal in the size of the output. The matrices and vectors used by this process can be reset to all zeros after processing each node: the total time spent in such reinitializations in $O(n)$. 

If we just need $\COMPLEXITY_{-}(T)$, rather than storing the temporary matrices $\mathtt{extensions}$ and $\mathtt{leftExtensions}$, we store just a number $\mathtt{area}$ which we initialize to $hk$ before processing node $v$. Whenever we observe a right extension $c_j$ of a string $a_{i}\ell(v)$, we decrease $\mathtt{area}$ by one. Before moving to the next node, we increment $\COMPLEXITY_{-}(T)$ by $\mathtt{area}$.
\end{proof}

Let $\mathbf{T}_{-}$ be the infinite-dimensional vector indexed by all distinct substrings $W \in [1..\sigma]^+$, such that $\mathbf{T}_{-}[W]=1$ iff $W$ is a minimal absent word of $T$. Theorem \ref{thm:maw} can be adapted to compute the Jaccard distance between the composition vectors of two strings:

\begin{corollary}\label{cor:jaccard}
Let $T^1 \in [1..\sigma]^{n_1-1}\#_1$ and $T^2 \in [1..\sigma]^{n_2-1}\#_2$ be strings. Given data structures that support $\mathtt{rangeDistinct}$ queries on $\BWT_{T^1}$ and on $\BWT_{T^2}$, respectively, we can compute $J(\mathbf{T}_{-}^1,\mathbf{T}_{-}^2)$ in $O(nd)$ time and in $o(n)$ bits of space in addition to the input, where $n=n_1+n_2$.
\end{corollary}
\begin{proof}
We apply the strategy of Theorem \ref{thm:maw} to the internal nodes of the generalized suffix tree of $T^1$ and $T^2$ whose label is a maximal repeat of $T^1$ and a maximal repeat of $T^2$: such strings are clearly maximal repeats of $T^{1}T^2$ as well. We enumerate such nodes as described in Section \ref{sec:enumerator}. We keep a global variable $\mathtt{intersection}$ and a bitvector $\mathtt{sharedRight}[1..\sigma]$. For every node $v$ that corresponds to a maximal repeat of $T^{1}$ and of $T^2$, we merge the sorted arrays $\mathtt{chars}_1$ and $\mathtt{chars}_2$ of $\repr'(\ell(v))$, we set $\mathtt{sharedRight}[c]=1$ for every character $c$ that belongs to the intersection of the two arrays, and we cumulate in a variable $k'$ the number of ones in $\mathtt{sharedRight}$. Then, we scan every left extension $a_i$ provided by $\mathtt{extendLeft}'$, we determine in constant time whether it occurs in both $T^1$ and $T^2$, and if so we increment a variable $h'$ by one. Finally, we initialize a variable $\mathtt{area}$ to $h'k'$, and we process again every left extension $a_i$ provided by $\mathtt{extendLeft}'$: if $a_{i}\ell(v)$ occurs in both $T^1$ and $T^2$, we compute the union of arrays $\mathtt{chars}_1$ and $\mathtt{chars}_2$ of $\repr'(a_{i}\ell(v))$, and for every character $c$ in the union such that $\mathtt{sharedRight}[c]=1$, we decrement $\mathtt{area}$ by one. At the end of this process, we add $\mathtt{area}$ to the global variable $\mathtt{intersection}$. To compute $||\mathbf{T}_{-}^1 \vee \mathbf{T}_{-}^2||$ we apply Theorem \ref{thm:maw} to $T^1$ and $T^2$ separately.
\end{proof}

It is easy to extend Corollary \ref{cor:jaccard} to compute $\kappa(\mathbf{T}_{-}^1,\mathbf{T}_{-}^2)$, as well as to support weighting schemes based on the length and on the characters of minimal absent words.

%

\section{Markovian corrections}\label{sec:markovianCorrection}

In some applications it is desirable to assign to component $W \in [1..\sigma]^k$ of composition vector $\mathbf{T}_{\infty}$ an estimate of the \emph{statistical significance} of observing $f_{T}(W)$ occurrences of $W$ in $T$: intuitively, strings whose frequency departs from its expected value are more likely to carry ``information'', and they should be weighted more \cite{qi2004whole}. Assume that $T$ is generated by a Markov random process of order $k-2$ or smaller, that produces strings on alphabet $[1..\sigma]$ according to a probability distribution $\mathbb{P}$. It is well known that the probability of observing $W$ in a string generated by such random process is $\mathbb{P}(W) = \mathbb{P}(W[1 \ltdots k-1]) \cdot \mathbb{P}(W[2 \ltdots k]) / \mathbb{P}(W[2 \ltdots k-1])$. We can estimate $\mathbb{P}(W)$ using the empirical probability $p_{T}(W)$, obtaining the following approximation for $\mathbb{P}(W)$: $\tilde{p}_{T}(W) = p_{T}(W[1 \ltdots k-1]) \cdot p_{T}(W[2 \ltdots k]) / p_{T}(W[2 \ltdots k-1])$ if $p_{T}(W[2 \ltdots k-1]) \neq 0$, and $\tilde{p}_{T}(W) = 0$ otherwise. We can thus estimate the significance of the event that substring $W$ has empirical probability $p_{T}(W)$ in string $T$ using the following score: $z_{T}(W) = ( p_{T}(W)-\tilde{p}_{T}(W) ) / \tilde{p}_{T}(W)$ if $\tilde{p}_{T}(W) \neq 0$, and $z_{T}(W)=0$ if $\tilde{p}_{T}(W)=0$ \cite{qi2004whole}. After elementary manipulations \cite{denas}, $z_{T}(W)$ becomes: 
\begin{eqnarray*}
z_{T}(W) & = & g(n,k) \cdot \frac{f_{T}(W) \cdot f_{T}(W[2 \ltdots k-1])}{f_{T}(W[1 \ltdots k-1]) \cdot f_{T}(W[2..k])} - 1 \\
g(x,y) & = & (x-y+2)^2 / (x-y+1)(x-y+3)
\end{eqnarray*}
Since $g(x,y) \in [1..1.125]$, we temporarily assume $g(x,y)=1$ in what follows, removing this assumption later.

Let $\mathbf{T}_z$ be a version of the infinite-dimensional vector $\mathbf{T}_{\infty}$ in which $\mathbf{T}_{z}[W]=z_{T}(W)$. Among all strings that occur in $T$, only strings $aWb$ such that $a$ and $b$ are characters in $[0..\sigma]$ and such that $W$ is a maximal repeat of $T$ can have $\mathbf{T}_z[aWb] \neq 0$. Similarly, among all strings that \emph{do not occur} in $T$, only the minimal absent words of $T$ have a nonzero component in $\mathbf{T}_z$: specifically, $\mathbf{T}_z[aWb]=-1$ for all minimal absent words $aWb$ of $T$, where $a$ and $b$ are characters in $[0..\sigma]$ \cite{denas}. Given two strings $T^1$ and $T^2$, we can thus compute $\kappa(\mathbf{T}^1_z,\mathbf{T}^2_z)$ using the same strategy as in Corollary \ref{cor:jaccard}:

\begin{theorem}\label{thm:markovian}
Let $T^1 \in [1..\sigma]^{n_1-1}\#_1$ and $T^2 \in [1..\sigma]^{n_2-1}\#_2$ be strings. Given data structures that support $\mathtt{rangeDistinct}$ queries on $\BWT_{T^1}$ and on $\BWT_{T^2}$, respectively, and assuming $g(x,y)=1$ for all settings of $x$ and $y$, we can compute $\kappa(\mathbf{T}^1_z,\mathbf{T}^2_z)$ in $O(nd)$ time and in $o(n)$ bits of space in addition to the input, where $n=n_1+n_2$.
\end{theorem}
\begin{proof}
We focus here on computing component $N$ of $\kappa(\mathbf{T}^1_z,\mathbf{T}^2_z)$: computing $D^i$ follows a similar algorithm on $\BWT_{T^i}$. We keep again a bitvector $\mathtt{sharedRight}[1..\sigma]$, and we enumerate all the internal nodes of the generalized suffix tree of $T^1$ and $T^2$ whose label is a maximal repeat of $T^{1}T^2$, as described in Section \ref{sec:enumerator}. For every such node $v$, we merge the sorted arrays $\mathtt{chars}_1$ and $\mathtt{chars}_2$ of $\repr'(\ell(v))$, we set $\mathtt{sharedRight}[c]=1$ for every character $c$ that belongs to the intersection of the two arrays, and we cumulate in a variable $k'$ the number of ones in $\mathtt{sharedRight}$. Then, we scan every left extension $a_i$ provided by $\mathtt{extendLeft}'$, we determine in constant time whether it occurs in both $T^1$ and $T^2$, and if so we increment a variable $h'$ by one. Finally, we initialize a variable $\mathtt{area}$ to $h'k'$, and we process again every left extension $a_i$ provided by $\mathtt{extendLeft}'$. If $a_{i}\ell(v)$ occurs in both $T^1$ and $T^2$, we merge arrays $\mathtt{chars}_1$ and $\mathtt{chars}_2$ of $\repr'(a_{i}\ell(v))$: for every character $b$ in the intersection of $\mathtt{chars}_1$ and $\mathtt{chars}_2$, we add to $N$ value $z_{1}(a_{i}\ell(v)b) \cdot z_{2}(a_{i}\ell(v)b)$, retrieving the corresponding frequencies from $\repr'(a_{i}\ell(v))$ and from $\repr'(\ell(v))$, and we decrement $\mathtt{area}$ by one. For every character $b$ that occurs only in $\mathtt{chars}_1$, we test whether $\mathtt{sharedRight}[b]=1$: if so, $a_{i}Wb$ is a minimal absent word of $T^2$ that occurs in $T^1$, thus we decrement $\mathtt{area}$ by one and we add to $N$ value $-z_{1}(a_{i}\ell(v)b)$. We proceed symmetrically if $b$ occurs only in $\mathtt{chars}_2$. At the end of this process, $\mathtt{area}$ counts the number of minimal absent words with infix $\ell(v)$ that are shared by $T^1$ and $T^2$: thus, we add $\mathtt{area}$ to $N$.
\end{proof}

It is easy to remove the assumption that $g(x,y)$ is always equal to one. There are only two differences from the previous case. First, the score of the substrings $W$ of $T^i$ that have a maximal repeat of $T^i$ as an infix changes, but $g(n_i,|W|)$ can be immediately computed from $|W|$, which is included in both $\repr(W)$ and $\repr'(W)$. Second, the score of all substrings $W$ of $T^i$ that do not have a maximal repeat as an infix changes from zero to $g(n_i,|W|)-1$: we can take into account all such contributions by pushing prefix-sums to the stack, as in Corollary \ref{cor:weightedSubstringKernel}. For example, to compute component $N$ of $\kappa(\mathbf{T}^1_z,\mathbf{T}^2_z)$, we can first assume that \emph{all} substring $W$ that occur both in $T^1$ and in $T^2$ have score $g(n_i,|W|)-1$, by pushing to the stack the prefix-sums described in \cite{denas}  and by enumerating only nodes $v$ of the generalized suffix tree of $T^1$ and $T^2$ such that $\ell(v)$ occurs both in $T^1$ and in $T^2$. Then, we can run the algorithm in Theorem \ref{thm:markovian}, subtracting quantity $(g(n_1,|W|+2)-1) \cdot (g(n_2,|W|+2)-1)$ from the contribution to $N$ of every string $a_{i}Wb$ that occurs both in $T^1$ and in $T^2$. 

Finally, recall that in Section \ref{sect:kmerKernels} we mentioned the problem of determining an \emph{upper bound} on the values of $k$ to be used in $k$-mer kernels. Let $\mathbf{T}_k$ be the composition vector indexed by all strings in $[1..\sigma]^k$ such that $\mathbf{T}_k[W]=p_{T}(W)$, and let $\tilde{\mathbf{T}}_k$ be a similar composition vector with $\tilde{\mathbf{T}}_k[W]=\tilde{p}_{T}(W)$, where $\tilde{p}_{T}(W)$ is defined as in the beginning of this section. 
It makes sense to disregard values of $k$ for which $\mathbf{T}_k$ and $\tilde{\mathbf{T}}_k$ are very similar, and more formally whose Kullback-Leibler divergence $\KL(\mathbf{T}_k,\tilde{\mathbf{T}}_k) = \sum_{W}\mathbf{T}_{k}[W] \cdot ( \log(\mathbf{T}_{k}[W])-\log(\tilde{\mathbf{T}}_{k}[W]) )$ is small, where $W$ ranges over all strings in $[1..\sigma]^k$. Thus, we could use as an upper bound on $k$ the minimum value $k^*$ such that $\sum_{k'=k^*}^{\infty}\KL(\mathbf{T}_{k'},\tilde{\mathbf{T}}_{k'}) < \tau$ for some user-specified threshold $\tau$ \cite{sims2009alignment}. Note again that only strings $aWb$ such that $a$ and $b$ are characters in $[0..\sigma]$ and $W$ is a maximal repeat of $T$ contribute to $\KL(\mathbf{T}_{|W|+2},\tilde{\mathbf{T}}_{|W|+2})$. We can thus adapt Theorem \ref{thm:markovian} to compute the KL divergence for a user-specified \emph{range of lengths} $[k_1..k_2]$, using just one pass over $\BWT_T$, in $O(nd)$ time and in $o(n)$ bits of space in addition to the input and the output. The same approach can be used to compute the \emph{KL-divergence kernel} $\kappa(\mathbf{T}_{KL}^1,\mathbf{T}_{KL}^2)$, where $\mathbf{T}_{KL}^{i}[W]=\KL_{T^i}(W)$ and $\KL_{T^i}(W) = \sum_{a,b \in \Sigma}p_{T^i}(aWb) \cdot ( \log(p_{T^i}(aWb)) - \log(\tilde{p}_{T^i}(aWb)) )$.

\bibliographystyle{plain}
\bibliography{string_analysis}
\end{document}